\documentclass[a4paper,UKenglish]{llncs}

\usepackage{microtype}
\usepackage{bst}

\title{Parameterized Enumeration with Ordering\thanks{%
Supported by a Campus France/DAAD Procope grant, Campus France Projet No~28292TE, DAAD Projekt-ID 55892324.}}
\titlerunning{Parameterized Enumeration with Ordering}

\author{Nadia~Creignou\inst{1} \and Ra\"\i da~Ktari\inst{1} \and Arne~Meier\inst{2} \and Julian-Steffen~M\"uller\inst{2} \and Frédéric~Olive\inst{1} \and Heribert~Vollmer\inst{2}}

\institute{Aix-Marseille Université, CNRS LIF UMR 7279 \\ 
\texttt{$\{$Nadia.Creignou, Frederic.Olive, Raida.Ktari$\}$@lif.univ-mrs.fr} \and
Leibniz Universit\"at Hannover \\ 
\texttt{$\{$meier, mueller, vollmer$\}$@thi.uni-hannover.de}}

\authorrunning{N.~Creignou, R.~Ktari, A.~Meier, J.-S.~M\"uller, F.~Olive, and H.~Vollmer}

\begin{document}
\maketitle

\begin{abstract}
The classes $\delayFPT$ and $\totalFPT$ recently have been introduced into parameterized complexity in order to capture the notion of efficiently solvable parameterized enumeration problems. In this paper we focus on ordered enumeration and will show how to obtain $\delayFPT$ and $\totalFPT$ enumeration algorithms for several important problems. We propose a generic algorithmic strategy, combining well-known principles stemming from both parameterized algorithmics and enumeration, which shows that, under certain preconditions, the existence of a so-called neighbourhood function among the solutions implies the existence of a $\delayFPT$ algorithm which outputs all ordered solutions. In many cases, the cornerstone to obtain such a neighbourhood function is a $\totalFPT$ algorithm that outputs all minimal solutions. This strategy is formalized in the context of graph modification problems, and shown to be applicable to numerous other kinds of problems.
\end{abstract}

\section{Introduction}
Enumeration problems, the task of generating all solutions of a given computational problem, find applications, e.\,g., in query answering in databases and web search engines, bioinformatics and computational linguistics. From a complexity-theoretic viewpoint, the notions of Total-P, the class of those problems whose solutions can be output in time polynomial in the input length, and of Delay-P, the class of problems that can be output in such a way that the delay between two outputs is bounded by a polynomial, are of utmost importance \cite{JohnsonPY88}. 

For many enumeration problems it is of high interest that the output solutions obey some given ordering. In particular this is interesting as you then get the minimal solutions (e.\,g., the cost optimal solutions or more generally the most important solutions in some metric) at the beginning of the enumeration algorithm. 
Let us illustrate this with some examples. 

The question for which classes of propositional CNF formulas an enumeration of all satisfying solutions is possible in Delay-P was studied in \cite{CreignouH97}. In terms of the well-known Schaefer framework for classification of Boolean constraint satisfaction problems, it was shown that for the classes of Horn, anti-Horn, affine or bijunctive formulas, such an algorithm exists. For other classes of formulas the existence of a Delay-P algorithm implies $\PTime=\NP$. It is interesting to note that the result hinges on the self-reducibility of the propositional satisfiability problem. Since variables systematically are tried first with an assignment $0$ and then $1$, it can be observed that the given enumeration algorithms output all satisfying assignments in lexicographic order. 

In \cite{CreignouOS11} the enumeration of satisfying assignments for propositional formulas was studied under a different order, namely in non-decreasing weight, and it was shown that under this new requirement, enumeration with polynomial delay is only possible for Horn formulas and width-2 affine formulas (i.e., affine formulas with at most 2 literals per clause). One of the main ingredients of these algorithms is the use of a priority queue to ensure enumeration in order (as is the case already in \cite{JohnsonPY88}).

Recently, analogues of the just mentioned classes of enumeration problems in the context of parameterized complexity have been introduced under the names $\totalFPT$ and $\delayFPT$ \cite{cmmsv13}. The ``polynomial time'' in the former definitions here is simply replaced by a time-bound of the form $p(n)\cdot f(k)$, where $n$ denotes the input length and $k$ its parameter, $p$ is an arbitrary polynomial, and $f$ is an arbitrary recursive function. By this the notion of efficiency in the context of the parameterized world, i.e., fixed parameter tractability ($\FPT$), has been combined with the enumeration framework. A number of problems from propositional logic were studied in \cite{cmmsv13} and enumeration algorithms based on self-reducibility and on the technique of kernelization were developed. In particular it was shown that membership of an enumeration problem in $\delayFPT$ can be characterized by a certain tailored form of kernalizability, very much as in the context of usual decision problems.

In the present paper we study \emph{ordered enumeration} in the context of \emph{parameterized complexity}. We present a generic algorithm proving that parameterized ordered enumeration problems are in $\delayFPT$ as soon as a certain
\FPT-computable neighbourhood function on the solution set exists (see \Cref{thm:Nb in FPT -> allMP in delayFPT}).
Computation of the neighbourhood function often starts with an initial phase generating all minimal solutions in $\totalFPT$ (in arbitrary order). For many problems this can quite easily be achieved by a bounded search tree algorithm, e.\,g., the well-studied so-called \emph{graph modification problems} with finite forbidden pattern characterization.
In a second step, we prove that neighbourhood functions can always be obtained by an iterative application of a $\totalFPT$ algorithm to enumerate all minimal solutions (see \Cref{thm:main}).
We prove the wide scope of applications of our method by presenting $\FPT$-delay ordered enumeration algorithms for a large variety of problems, such as cluster editing, chordal completion, closest-string, and weak and strong backdoors.

\section{Preliminaries}

We start by defining parameterized enumeration problems with a specific ordering and their corresponding enumeration algorithms.
Most definitions in this section transfer those of \cite{JohnsonPY88,Schmidt09} from the context of enumeration and those of \cite{cmmsv13} from the context of parameterized enumeration to the context of parameterized \emph{ordered} enumeration. 

%

\begin{definition}\label[definition]{def:para-enum-pb}
	A \emph{parameterized enumeration problem with ordering} is a quadruple $E=(I, \kappa, \Sol, \preceq)$ such that the following holds:
	\begin{itemize}
		\item	$I$ is the set of \emph{instances}.
		\item $\kappa\colon I\rightarrow \N$ is the \emph{parameterization function};
			$\kappa$ is required to be polynomial-time computable.
		\item $\Sol$ is a function such that for all $x\in I$, 
				$\Sol(x)$ is a finite set, the set of \emph{solutions} of $x$.
			Further we write $\AllSol=\bigcup_{x\in I}\Sol(x)$.
		\item $\preceq$ is a quasiorder (or \emph{preorder}, i.\,e., a reflexive and transitive binary relation) on $\AllSol$.
	\end{itemize}
\end{definition}

We will write $I_{E}$, $\kappa_{E}$, etc.\ to denote that we talk about instance set, parameterization function, etc.\ of problem $E$.

\begin{definition}
	Let $E=(I,\kappa,\Sol,\preceq)$ be a parameterized enumeration problem with ordering. Then an algorithm $\calA$ is an \emph{enumeration algorithm} if the following holds:
	\begin{itemize}
		\item For every $x\in I$, $\calA(x)$ terminates after a finite number of steps.
		\item For every $x \in I$, $\calA(x)$ outputs exactly the elements of $\Sol(x)$ without duplicates.
		\item For every $x\in I$ and $y,z\in\Sol(x)$, if $y \preceq z$ then $\calA(x)$ outputs solution $y$ before solution $z$.
	\end{itemize}
\end{definition}

Before we define complexity classes for parameterized enumeration, we need the notion of delay for enumeration algorithms.
 
\begin{definition}[Delay]
 Let $E=(I,\kappa,\Sol, \preceq)$ be a parameterized enumeration problem with ordering and $\calA$ be an enumeration algorithm for $E$.
Let $x\in I$ be an instance. Then we say that the $i$-th delay of $\calA$ is the time between outputting the $i$-th and $(i+1)$-st solution in $\Sol(x)$. Further, we define the $0$-th delay to be the \emph{precalculation time} which is the time from the start of the computation to the first output statement. Analogously, the $n$-th delay, for $n=|\Sol(x)|$, is the \emph{postcalculation time} which is the time needed after the last output statement until $\mathcal A$ terminates.
\end{definition}

Now we are able to define two different complexity classes for parameterized enumeration following the notion of \cite{cmmsv13}.

\begin{definition}
Let $E=(I, \kappa, \Sol, \preceq)$ be a parameterized enumeration problem and $\calA$ an enumeration algorithm for $E$.
\begin{enumerate}
\item The algorithm $\calA$ is a $\totalFPT$ algorithm if there exist a computable function $t\colon \N\rightarrow \N$ and a polynomial $p$ such that for every instance $x\in\Sigma^*$, $\mathcal{A}$ outputs all solutions of $\Sol(x)$ in time at most $t(\kappa(x))\cdot p(|x|)$.
\item The algorithm $\calA$ is a $\delayFPT$ algorithm if there exist a computable function $t\colon \N\rightarrow \N$ and a polynomial $p$ such that for every $x\in\Sigma^*$, $\mathcal{A}$ outputs all solutions of $\Sol(x)$ with delay of at most $t(\kappa(x))\cdot p(|x|)$.
\end{enumerate}
\end{definition}

\begin{definition}
The class $\totalFPT$ (resp., $\delayFPT$)
is the class of all parameterized enumeration problems that admit a $\totalFPT$ (resp., $\delayFPT$)
enumeration algorithm.  
\end{definition}

Whenever we obtain $\delayFPT$ algorithms we will make inherent use of the concept of \emph{priority queues} to enumerate all solutions in the correct order and to avoid duplicates. We will follow the approach of Johnson et.~al.\ \cite{JohnsonPY88}. A priority queue $Q$ stores a potentially exponential number of elements. Let $x$ be an instance. The \emph{insert operation} of $Q$ requires $O(|x|\cdot\log |\Sol(x)|)$. This cost is also charged while avoiding the insertions of an element that already is in $Q$. The \emph{extract minimum operation} requires $O(|x|\cdot\log |\Sol(x)|)$ time, too. It is important, however, that the computation of the order between two elements takes at most $O(|x|)$ time. In this paper we will focus on solutions which are sets (of certain graph operations) and these sets will always be ordered by the quasiorder of non-decreasing cardinality. This quasiorder can easily be extended to a total order by ordering sets of equal cardinality lexicographically (fixing some order among the elements). 

As pointed out by Johnson et.\ al.\ this type of queue can be implemented with the help of standard balanced tree schemes.

\section{Graph Modification Problems}

Graph modifications problems have been studied for a long time in computational complexity theory. Already in the monograph by Garey and Johnson \cite{gj90}, among the graph-theoretic problems considered, many fall into this problem class. To the best of our knowledge, graph modification problems were studied in the context of parameterized complexity for the first time in \cite{cai96}. Given some graph property $\mathcal P$ and some graph $G$, we write $G\models\mathcal P$ if the graph $G$ obeys the property $\mathcal P$. A \emph{(graph) operation} for $G$ is either removing a vertex, or adding/removing an edge. A set of operations is \emph{consistent} if it does not contain an operation which adds and removes the same edge, and if it does not contain an operation removing a vertex and adding/removing edges incident to this vertex. Given such a set of consistent operations $S$, the graph obtained from $G$ in applying the operations in $S$ is denoted by $S(G)$. 

\begin{definition}
 Given some graph property $\mathcal P$, a graph $G$ and $k\in\N$, we say that $S$ is a \emph{solution for $(G,k)$} with respect to $\calP$ if the following three properties hold:
\begin{enumerate}
 \item $S$ is a consistent set operations.
 \item $|S|\leq k$.
 \item $S(G)\models\calP$.
\end{enumerate}
A solution $S$ is \emph{minimal} if for all $S'$ with $S'\subsetneq S$ it holds that $S'$ is not a solution.
\end{definition}  

Cai was interested in the following graph modification decision problem w.r.t.\ some given graph property $\calP$, $\MP$: given a graph $G$ and an integer $k$, does there exist a solution for $(G,k)$?
%
%
Here we will focus on the following two enumeration problems.

\enumproblem{\minMP}
{$(G,k)$, where $G=(V,E)$ is an undirected graph and $k\in\N$.}
{$k$}
{Generate all minimal solutions $S$ by non-decreasing size.}

We denote the corresponding problem to enumerate \emph{all} solutions (not only minimal ones) by non-decreasing size by $\allMP$.
%

In other words, the set of instances $I_{\MP}$ (in the sense of \Cref{def:para-enum-pb}) is the set of all pairs $x=(G,k)$, and $\kappa_{\MP}(G,k)=k$.
If the context is clear we omit the subscript $\mathcal P$ for the graph modification problem and simply write $\mathcal M$. We write $\Sol_\calM(x)$ for the function associating solutions to a given instance, and also $\AllSol_\calM$ for the set of all solutions of $\calM$. 

Cai studied the parameterized complexity of $\MP$ and obtained a positive result. In order to state it let us introduce some terminology. Given two graphs $G=(V,E)$ and $H=(V',E')$, we write $H\leq G$ if $H$ is an induced subgraph of $G$, i.\,e., $V'\subseteq V$ and $E'=E\cap (V'\times V')$.
Let $\mathcal F$ be a set of graphs and $\mathcal P$ be some graph property. We say, $\mathcal P$ \emph{has a finite forbidden set characterization} $\mathcal F$ if for any graph $G$ it holds that $G\models\mathcal P$ if and only if there exists no induced subgraph $H\leq G$ with $H\in\mathcal F$, and $|\mathcal F|<\infty$. 

\begin{proposition}[\cite{cai96}]
 The problem $\MP$ is in $\FPT$ for any property $\calP$ that has a finite forbidden set characterization. 
\end{proposition}

The algorithm Cai developed for the decision problem is based on a bounded search tree, whose exhaustive examination provides all minimal solutions. Thus we get the following. 

\begin{lemma}\label[lemma]{prop:minMP}
 The problem $\minMP$ is in $\totalFPT$ for any property $\calP$ that has a finite forbidden set characterization. 
\end{lemma}

\begin{proof}
Let $\mathcal P$ be the given graph property and $\mathcal F$ be its finite forbidden set characterization. The constructed bounded search tree algorithm is depicted in \Cref{alg:bst-gp}.

\begin{algorithm}
 \caption{Enumeration algorithm for $\minMP\in\totalFPT$} \label[algorithm]{alg:bst-gp}
 \LinesNumberedHidden
 \Input{$(G,k)$, where $G$ is an undirected graph and $k\in\N$}
 \ShowLn priority queue $Q\leftarrow\emptyset$ (ordered by size)\;
 \ShowLn {\ttfamily \textsc{Min-BST}($G,k,\emptyset,Q$)}\;
 \ShowLn \lWhile{$Q$ is not empty}{
 \textbf{output} head element of $Q$\;
 }
 \BlankLine\BlankLine
 \MinBST{Undirected graph $G$, integer $k\in\N$, set of operations $S$, queue $Q$}\\\setcounter{AlgoLine}{0}
 \ShowLn \lIf{$S(G)\models\calP$}{
   insert $S$ into $Q$ and \textbf{halt}\;
 }
 \ShowLn \lIf{$|S|=k$}{
   \textbf{halt}\;
  }
 \ShowLn search an induced subgraph $H\leq S(G)$ such that $H\in\mathcal F$\;
 \ShowLn \lForAll{graph operations $t$ consistent with $S$}{
    {\ttfamily \textsc{Min-BST}($G,k,S\cup\{t\},Q$)}\;
 }
\end{algorithm}

Since $\mathcal F$ is finite, the size of any $|H|$ in $\mathcal F$ is bounded by some $c\in\N$. The branching width of the search tree is bounded $c+c\cdot (c-1)$. The depth of the recursion is bounded by $k$ whence we get an overall tree size of $(c+c\cdot (c-1))^{k}$. The operations on the priority queue run in polynomial time which lead in combination with the polynomial time of the pattern search to overall $\totalFPT$.
\end{proof}

Observe that one cannot expect a similar result for \allMP. One can easily see that $\allMP\notin\totalFPT$ holds by a counting argument. Indeed the number of solutions is potentially too large to allow for a $\totalFPT$ algorithm. Enumerating all solutions by non-decreasing size can nevertheless be hopefully made in $\delayFPT$. This requires a more involved strategy that amounts to generate all solutions from an initial ``seed'' (denoted by $\seed$ below), as well as taking care of the order. The second point will be handled by a priority queue. For the first one we need some notion of operations which transform one solution into possible ``neighbour solutions'' (if such exist).

\begin{definition}
\label[definition]{def:nbf}
Let $\Mprob$ be some graph modification problem. A \emph{neighbourhood function} for $\Mprob$ is a (partial) function $\nbf{\Mprob}\colon I_{\Mprob}\times \left(\AllSol_{\Mprob}\cup\{\seed\}\right)\rightarrow\calP({\AllSol_{\Mprob}})$
such that the following holds:
\begin{enumerate}
\item For all $x=(G,k)\in I_{\Mprob}$ and $S\in\Sol_{\Mprob}(x)\cup\{\seed\}$, 
		$\nbf\Mprob(x,S)$ is defined.
\item\label{item:nbf init}	For all $x\in I_{\Mprob}$, 
		$\nbf\Mprob(x,\seed)=\emptyset$ if $\Sol_{\Mprob}(x)=\emptyset$, 
		and $\nbf\Mprob(x,\seed)$ is an arbitrary set of solutions otherwise.		
\item\label{item:nbf extend} For all $x\in I_{\Mprob}$ and $S\in\Sol_{\Mprob}(x)$, 
		if $S'\in\nbf\Mprob(x,S)$ then $|S|<|S'|$.  
\item\label{item:nbf cover} For all $x\in I_{\Mprob}$ and all $S\in\Sol_{\Mprob}(x)$, 
		there exist $p>0$ and $S_{1},\dots,S_{p}\in\Sol_{\Mprob}(x)$ 
			such that 
\begin{itemize}
\item $S_{1}\in\nbf{\Mprob}(x,\seed)$,
\item $S_{i+1}\in\nbf{\Mprob}(x,S_{i})$ for $1\leq i<p$, and 
\item $S_{p}=S$.
\end{itemize}
\end{enumerate}
Furthermore, we say that $\nbf{\Mprob}$ is $\FPT$-computable, when $\nbf{\Mprob}(x,S)$ is computable in time $f(k)\cdot\text{poly}(|x|)$ for any $x\in I_{\Mprob}$ and $S\in\Sol_{\Mprob}(x)$.
\end{definition}

Thus, a neighbourhood function for a problem $\Mprob$ is a function that in a first phase computes from scratch some initial set of solutions (see \Cref{def:nbf}(\ref{item:nbf init})). 
In many of our applications below, $\nbf\Mprob(x,\seed)$ will be the set of all minimal solutions for $x$.
In a second phase these solutions are iteratively enlarged
(see condition~(\ref{item:nbf extend})), 
where
condition~(\ref{item:nbf cover})
guarantees that we do not miss any solution, as we will see in the next theorem.

\begin{theorem}\label[theorem]{thm:Nb in FPT -> allMP in delayFPT}
Let $\calM$ be a graph modification problem. If $\calM$ admits a neighbourhood function $\nbf\calM$ that is \FPT-computable, then $\allM$ is in $\delayFPT$.
\end{theorem}

\begin{proof}
The algorithm which outputs all solutions in $\delayFPT$ is shown in \Cref{alg:delayfpt-withNBF}. 
By the definition of the priority queue (recall in particular that insertion of an element is done only if the element is not yet present in the queue) and by the fact that all elements of $\nbf\calM((G,k),S)$ are of bigger size than $S$ by \Cref{def:nbf}(\ref{item:nbf extend}), it is easily seen that \emph{the solutions are output in the right order} and that \emph{no solution is output twice}.

Besides, \emph{no solution is omitted}. Indeed, given $S\in\Sol_{\calM}(G,k)$ and $S_{1},\dots,S_{p}$ associated with $S$ by \Cref{def:nbf}(\ref{item:nbf cover}), we prove by induction that each $S_{i}$ is inserted in $Q$ during the run of the algorithm: For $i=1$, this proceeds from line~2 of the algorithm; for $i>1$, the solution $S_{i-1}$ is inserted in $Q$ by induction hypothesis and hence all elements of $\nbf{\calM}((G,k),S_{i-1})$, including $S_{i}$, are inserted in $Q$ (line~6). Thus, each $S_{i}$ is inserted in $Q$ and then output during the run. In particular, this holds for $S=S_{p}$.

Finally, we claim that \emph{\Cref{alg:delayfpt-withNBF} runs in $\delayFPT$}. Indeed, the delay between the output of two consecutive solutions is bounded by the time required to compute a neighbourhood of the form $\nbf{\calM}((G,k),\seed)$ or $\nbf{\calM}((G,k),S)$ and to insert all its elements in the priority queue. This is in $\FPT$ due to the assumption on $\nbf{\calM}$ being $\FPT$-computable and as there is only a single extraction and $\FPT$-many insertion operations on the queue.
%
\begin{algorithm}
\caption{$\delayFPT$ algorithm for $\allM$}\label[algorithm]{alg:delayfpt-withNBF}
\Input{$(G,k),$ where $G$ is an undirected graph and $k\in\N$.}
compute $\nbf\calM((G,k),\seed)$\;
insert all elements of $\nbf\calM((G,k),\seed)$ into the priority queue $Q$ (ordered by size)\;
\While{$Q$ is not empty}{
	\textbf{extract} the minimum solution $S$ of $Q$\;
	\textbf{output} $S$\;
	\textbf{insert} all elements of $\nbf\calM((G,k),S)$ into $Q$\; 
}
\end{algorithm}
\end{proof}

\begin{figure}
 \centering
\usetikzlibrary{arrows}
\begin{tikzpicture}[box/.style={rectangle,draw,black,text width=2.5cm,align=center},x=0.9cm]
 \node[box,text width=2cm] (BST) at (0,0) {$\nb{}((G,k),\seed)$};
 \node[box,text width=2cm] (pq) at (5,0) {priority queue};
 \node[box,dashed] (x) at (10,0) {output current solution $S$};
 
 \path[-triangle 45,draw,black] (BST) edge node [above] {initial} (pq);
 \path (BST) edge node [below] {solutions} (pq);
 \path[-triangle 45,draw,black] (pq) edge node [above] {extract} (x);
 \path (pq) edge node [below] {head} (x);
 \path[-triangle 45,draw,black] (x) edge [bend left,looseness=0.8] node [below] {insert $\nb{}((G,k),S)$} (pq);
\end{tikzpicture}
 \caption{Structure of \Cref{alg:delayfpt-withNBF}.}
\end{figure}


Let us illustrate the application of the theorem to a specific modification problem, namely \CE. A \emph{cluster} is a graph such that all its connected components are cliques. In order to transform (or modify) a graph $G$ we allow here only two kinds of operations: adding or removing an edge. Given a graph $G$ and a parameter $k$, the question is whether there exists a consistent set of operations of cardinality at most $k$ such that $S(G)$ is cluster.


Here we are interested in enumerating by non-decreasing cardinality either all such consistent sets of operations, \allCE, or only all (inclusion) minimal ones, \minCE.


\begin{lemma}\label[lemma]{lem:minCE-totalfpt}
 $\minCE\in\totalFPT$.
\end{lemma}
\begin{proof}
The property for a graph to be a cluster has a finite forbidden set characterization (the forbidden pattern is a path of length 2). So, the lemma follows from \Cref{prop:minMP}.
\end{proof}


\begin{theorem}\label[theorem]{prop:allCE}
 $\allCE\in\delayFPT$.
\end{theorem}
\begin{proof} 
According to \Cref{thm:Nb in FPT -> allMP in delayFPT} we only have to provide an appropriate neighbourhood function. In a first phase, given an instance $(G,k)$, one can compute all minimal solutions, so defining $\calN((G,k),\seed)$. According to \Cref{lem:minCE-totalfpt} this is computable in $\FPT$. Now, observe that any solution can be obtained from a minimal one by iteratively applying one of the the two following operations: $(i)$ merging two cliques, or $(ii)$ cutting a clique into two parts. Let us take a closer look at these two operations. Given an instance $(G,k)$ and a solution $S$, define $\nb{}((G,k),S)$ as the set of all consistent modification sets of size at most $k$ obtained from $S$ in either merging two existing cliques (that is in adding the required edges), or splitting an existing clique into two parts (that is in removing the required edges). From the above observation it follows that property (4) in \Cref{def:nbf} is satisfied. Moreover all solutions in $\nb{}((G,k),S)$ are supersets 
of $S$, hence property (3) in \Cref{def:nbf} is satisfied as well. Therefore $\nb{}((G,k),S)$ is indeed a neighbourhood function. Moreover, observe that:

\begin{enumerate}
\item Merging two cliques $C$ and $C'$ of size $i$ and $j$ requires the introduction of $i\cdot j$ new edges. Therefore, only pairs of cliques such that $i\cdot j\leq k$ (and even $i\cdot j\leq k-|S|$) have to be considered. 
\item Cutting a clique of size $\ell$ into two cliques requires the deletion of at least $\ell-1$ edges. Therefore only cliques of size $\le k$ (and even $\ell\leq k-|S|$) have to be considered. 
\end{enumerate}
As a consequence the neighbourhood function $\nb{}((G,k),S)$ is computable in $\FPT$, thus concluding the proof.
\end{proof}

However, as discussed in the next section, finding an appropriate neighbourhood function is not always so easy.

\section{A Generic Enumeration Algorithm for the Modification Problem}

We aim to provide a generic algorithm for the problem \allMP. According to \Cref{thm:Nb in FPT -> allMP in delayFPT} our main goal is to provide a strategy to obtain an $\FPT$-computable neighbourhood function. In order to do so let us first start with a specific example.
\medskip

A \emph{chord} in a graph $G=(V,E)$ is an edge between to vertices of a cycle $C$ in $G$ which is not part of $C$.
A given graph $G=(V,E)$ is \emph{chordal} (or \emph{triangular}) if each of its induced cycles of 4 or more nodes has a chord. 

\decisionproblem{\CC}
{$(G,k)$, where $G$ is an undirected graph and $k\in\N$.}
{$k$}
{Does there exists a set of at most $k$ edges such that adding this set of edges to $G$ makes it chordal?}

Yannakakis showed that the corresponding decision problem is $\NP$-complete \cite{yannakakis81}. Kaplan et.\ al.\ \cite{kashta99}, and independently Cai \cite{cai96} have shown that the parameterized problem is in $\FPT$.
For this problem, a solution is a set of edges which have to be added to the graph to make the graph chordal, and we can define as previously the two enumeration problem $\minCC$ and $\allCC$.

Observe that in this special case of the modification problem the underlying property $\calP$, ``to be chordal'', does not have a finite forbidden set characterization (since cycles of any length can be problematic). However, one can efficiently enumerate all minimal solutions as well. 

\begin{lemma}\label[lemma]{prop:minCC}
 $\minCC\in\totalFPT$.
\end{lemma}
\begin{proof}
 We say a \emph{$k$-triangulation} of a given graph $G=(V,E)$ is a set of edges $E'$ such that $G=(V,E\cup E')$ is triangular and $|E'|\leq k$. Kaplan et.\ al.\ have shown that all minimal $k$-triangulations can be output in time $O(2^{4k}\cdot|E|)$ for a given graph $G=(V,E)$ and $k\in\N$ \cite[Thm.~2.4]{kashta99}. Afterwards, of course, with a simple sorting algorithm we can output the minimal solutions by non-decreasing size in $\totalFPT$ time.
\end{proof}

Further we will prove that \emph{all} solutions can be output ordered by non-decreasing size with delay $\FPT$.

\begin{theorem}\label{prop:allCC}
 $\allCC\in\delayFPT$.
\end{theorem}
\begin{proof}
The method is described in \Cref{alg:allCC}, it is based on \Cref{thm:Nb in FPT -> allMP in delayFPT} and follows the scheme provided in the proof of \Cref{prop:allCE}. 
As in the proof of \Cref{prop:allCE}, given an instance $(G,k)$ we start with the set of all its minimal solutions, which is computable in $\FPT$ according to \Cref{prop:minCC}. Next, given a solution $S$ we have to define the neighbour solutions. This neighbourhood cannot be obtained so easily as in the \CE problem by simple modifications of $S$. The idea is to take as a neighbourhood for $S$ all minimal solutions among the ones that are a superset of $S$. This can be computed in $\FPT$ since the minimal chordal completions of the original graph augmented by $S$ and any other edge can be computed in $\FPT$ according to \Cref{prop:minCC} (see lines~\ref{line:6}--\ref{line:7} of \Cref{alg:allCC}). The fact that the neighbourhood function so defined satisfies condition~(\ref{item:nbf cover}) in \Cref{def:nbf} can be proven by induction.
\end{proof}

\begin{algorithm} 
 \Input{$(G,k)$, where $G=(V,E)$ is an undirected graph and $k\in\N$}
 insert all minimal $k$-triangulations of $G$ into the priority $Q$ (ordered by size)\;
 \While{$Q$ is not empty}{
  extract the minimum solution $S$ from the queue $Q$\;
  \textbf{output} $S$\;
  \ForAll{$u,v\in V$ with $\textrm{addEdge}(u,v)\notin S$\label{line:6}}{
   insert the unions of $S$ with all minimal $(k-|S|-1)$-triangulations of $(V,E\cup S\cup\{(u,v)\})$ into $Q$\;\label{line:7}
  }
 }
 \caption{$\delayFPT$ algorithm for $\allCC$}\label[algorithm]{alg:allCC}
\end{algorithm}

This idea can easily be generalized to any graph modification problem. Thus we get the following theorem.

\begin{theorem}\label[theorem]{thm:main}
 Let $\calM$ be a graph modification problem. If $\minM\in\totalFPT$ then $\allM\in\delayFPT$.
\end{theorem}

\begin{proof}
Let $\mathcal A$ be an algorithm showing $\minM\in\totalFPT$. Because of \Cref{thm:Nb in FPT -> allMP in delayFPT}, it is sufficient to build an \FPT neighbourhood function for $\calM$. For an instance $(G,k)$ of ${\calM}$ and for $S\in\Sol_{{\calM}}(G,k)\cup\{\seed\}$, we define $\nbf{{\calM}}((G,k),S)$ as the result of \Cref{alg:buildingNbgh}.

\begin{algorithm}
\lIf{$S=\seed$\label{compNbghLine1}}{\textbf{return} $\mathcal A(G,k)$\;}
	{
	$\res=\emptyset$ \;
	\ForAll{graph operations $t$}{
		\ForAll{$S'\in\mathcal A((S\cup\{t\})(G),k-|S|-1)$\label{compNbghLine4}}{
			\lIf{$S\cup S'\cup\{t\}$ is consistent}{
			$\res=\res\cup\{S\cup S'\cup\{t\}\}$\;\label{compNbghLine5}
			}
		}
	}
	\textbf{return }$\res$\;
	}
\caption{Procedure for computing $\nbf{{\calM}}((G,k),S)$}\label[algorithm]{alg:buildingNbgh}
\end{algorithm}

The function $\nbf{{\calM}}$ thus defined clearly fulfills conditions~\ref{item:nbf init} and \ref{item:nbf extend} of \Cref{def:nbf}. We prove by induction that it also satisfies condition~\ref{item:nbf cover} (that is, each solution $T$ of size $k$ comes with a sequence $T_1,\dots,T_p=T$ such that $T_{1}\in\nbf{\calM}((G,k),\seed)$ and $T_{i+1}\in\nbf{\calM}((G,k),T_{i})$ for each $i$). 
If $T$ is a minimal solution for $(G,k)$, then $T\in\nbf{{\calM}}((G,k),\seed)$ and the expected sequence $(T_{i})$ reduces to $T_{1}=T$. 
Otherwise, there exist an $S\in\Sol_{\calM}(G,k)$ and a non-empty set of transformations, say $S'\cup\{t\}$, such that $T=S\cup S'\cup\{t\}$ and there is no solution for $G$ between $S$ and $S\cup S'\cup\{t\}$. 
This entails that $S'$ is a minimal solution for $\big((S\cup\{t\})(G),k-|S|-1\big)$ and hence $T\in\nbf{{\calM}}((G,k),S)$ (see lines~\ref{compNbghLine4}--\ref{compNbghLine5} of \Cref{alg:buildingNbgh}). 
The conclusion follows from the induction hypothesis that guarantees the existence of solutions $S_1,\dots,S_q$ such that $S_{1}\in\nbf{\calM}((G,k),\seed)$, $S_{i+1}\in\nbf{\calM}((G,k),S_{i})$ and $S_{q}=S$. The expected sequence $T_{1},\dots,T_{p}$ for $T$ is nothing but $S_1,\dots,S_q,T$. To conclude, it remains to see that \Cref{alg:buildingNbgh} is \FPT which follows from $\minM\in\totalFPT$ (lines~\ref{compNbghLine1} and \ref{compNbghLine4}), and the fact that there are only polynomial many graph operations which have to be considered.
\end{proof}
%
%
%
According to \Cref{prop:minMP}, we get the following result as an immediate corollary.

\begin{corollary}
Let $\calM_\calP$ be a graph modification problem. If $\mathcal P$ has a finite forbidden set characterization then $\allMP\in\delayFPT$.
\end{corollary}

Graph modification problems form a rich class of graph-theoretic algorithmic problems. Our results show that as soon as such a problem can be characterized by a finite forbidden set, the corresponding enumeration problem is in $\delayFPT$. The same holds if a characterization by a finite set of forbidden pattern is not given but the task to produce all minimal solutions is in $\totalFPT$ by some other means. We only give one further example in this paper:
The \problemFont{Triangle-Deletion} problem asks the question whether a given graph can be transformed into a triangle-free graph by deletion of at most $k$ vertices. Forbidden patterns are obviously just triangles.

\begin{corollary}
$\problemFont{All-Triangle-Deletion}\in\delayFPT$.
\end{corollary}

Analogous results hold for the modification problems for many other classes of graphs, e.\,g., line graphs, claw-free graphs, Helly circular-arc graphs, comparability graphs, etc., see \cite{brlesp88}.

\section{Various Further Examples}

In this section we will show how the algorithmic strategy that has been defined and formalized in the context of graph modification can be of use for many other problems, coming from various combinatorial frameworks.

\subsection{Closest-String}
Given a set of binary strings $I$ we want to find the string $s$ whose maximum Hamming distance to all elements in $I$ is at most $d$ for some given $d\in\N$. Frances and Litman have shown that the corresponding decision problem is \NP-hard \cite{frli97}. The problem is very important in the areas of census word analysis and computational biology. 

%
%

Given a string $w=w_{1}\cdots w_{n}$ with $w_{i}\in\{0,1\},n\in\N$, and a set $S\subseteq\{1,\dots,n\}$, $S(w)$ denotes the string obtained from $w$ in flipping the bits indicated by $S$, more formally $S(w):=S(w_{1})\cdots S(w_{n})$, where $S(w_{i})=1-w_i$ if $i\in S$ and $S(w_{i})=w_{i}$ otherwise.

The parameterized problem which will be the focus of this section is the following.
\decisionproblem
 {\CS}
 {$(s_1,\dots,s_k,n,d)$, where $s_1,\dots,s_k$ is a sequence of strings over $\{0,1\}$ of length $n\in\N$, and an integer $d\in\N$.}
 {$d$}
 {Does there exist $S\subseteq\{1,\dots,n\}$ such that $d_{H}(S(s_{1}),s_{i})\leq d$ for all $1\leq i \leq k$, where $d_{H}$ is the Hamming distance function?}

Gramm et.\ al.\ have shown that this problem is in $\FPT$ \cite{grniro03}. We are interested in the corresponding enumeration problems $\minCS$ and $\allCS$.

As $\CS$ is not a problem defined over graphs we cannot immediately use \Cref{prop:minMP} to obtain a $\totalFPT$ result for the minimal problem $\minCS$. Nevertheless we can achieve similarly this upper bound by an exhaustive examination of a bounded search tree which is constructed from the idea of Gramm et.\ al. \cite[Fig.\ 1]{grniro03}. This tree contains all minimal solutions and can be built in $\totalFPT$.

\begin{lemma}\label[lemma]{thm:min-closest-string-totalfpt}
 $\minCS\in\totalFPT$.
\end{lemma}
%
%
%

\begin{theorem}\label{thm:closest-string-delayfpt}
 $\allCS\in\delayFPT$.
\end{theorem}
\begin{proof}
The structure of the $\delayFPT$-algorithm is the same as the one developed in the proof of \Cref{thm:Nb in FPT -> allMP in delayFPT}. It uses an appropriate neighbourhood function, which as a first step provides minimal solutions and then iteratively produces all the remaining ones. A priority queue is used in order to avoid redundancy and to output the solutions in the desired order. 
Given an instance $x=(s_{1},\dots,s_{k},n,d)$ and $S\subseteq\{1,\dots,n\}$. The neighbourhood function $\nb{}(x,S)$ is then defined to be the set of the pairwise unions of all minimal solutions of $(S'(s_{1}),\dots,s_{k},d-|S|-1)$ which are disjoint from $S'$ together with $S'$, where $S':=S\cup\{i\}$ with $i\notin\{1,\dots,n\}\setminus S$.

%
%
\end{proof}

\subsection{Backdoors}
In the following, let $\calC$ be some class of CNF-formulas, and $\varphi$ be a propositional CNF formula. If $X$ is a set of propositional variables we denote with $\Theta(X)$ the set of all assignments over the variables in $X$. For some $\theta\in\Theta(X)$ the expression $\theta(\varphi)$ is the formula obtained by applying the assignment $\theta$ to $\varphi$, i.\,e., clauses with a satisfied literal are removed, and falsified literals are removed. A set $S$ of variables from $\varphi$ is a \emph{weak $\calC$-backdoor of $\varphi$} if there exists an assignment $\theta\in\Theta(S)$ such that $\theta(\varphi)\in\calC$ and $\theta(\varphi)$ is satisfiable. The set $S$ is a \emph{strong $\calC$-backdoor of $\varphi$} if for all $\theta\in\Theta(S)$ the formula $\theta(\varphi)$ is in $\calC$.

Now we can define the corresponding parameterized problems:

\decisionproblem
{Weak/strong-$\mathcal C$-Backdoors}
{A formula $\varphi$ in 3CNF, $k\in\N$.}
{$k$}
{Does there exist a weak/strong $\mathcal C$-backdoor of size $\leq k$?}

The class $\mathcal C$ is a \emph{base class} if it can be recognized in $\PTime$, satisfiability of its formulas is in $\PTime$, and the class is closed under isomorphisms w.r.t.\ variable names. We say $\mathcal C$ is \emph{clause defined} if for every CNF-formula $\varphi$ it holds: $\varphi\in\mathcal C$ iff $\{C\}\in\mathcal C$ for all clauses $C$ from $\varphi$. 
Gaspers and Szeider \cite{gasz12b} investigated a specific type of $\mathcal C$-formulas, namely the clause-defined base classes $\mathcal C$, and showed that for any such class $\mathcal C$, the detection of weak $\mathcal C$-backdoors is in $\FPT$ for input formulas in 3CNF. In the following we examine the parameterized enumeration complexity of the corresponding enumeration problems. They describe in \cite[Prop.~2]{gasz12b} that a bounded search tree technique allows to solve the detection of weak $\calC$-backdoors in $\FPT$ time. This technique results into obtaining all minimal solutions in $\FPT$ time and after the usual sorting we can enumerate them in $\totalFPT$.

\begin{lemma}
 For every clause-defined base class $\mathcal C$ and input formulas in 3CNF, $\problemFont{Min-Weak-$\mathcal C$-Backdoors}\in\totalFPT$.
\end{lemma}

\begin{theorem}
 For every clause-defined base class $\mathcal C$ and input formulas in 3CNF, $\problemFont{All-Weak-$\mathcal C$-Backdoors}\in\delayFPT$.
\end{theorem}
\begin{proof}
We proceed as before. Given an instance $x=(\varphi,k)$ we first compute all its minimal backdoors. Then, given some backdoor $S$ we define $\nb{}(x,S)$ as the set of the pairwise unions of all minimal weak $\calC$-backdoors of $(\theta(\varphi),k-|S|-1)$ together with $S\cup\{x_i\}$ for each $\theta\in\Theta(S\cup\{x_{i}\})$ for $x_{i}\in\textit{Vars}(\varphi)\setminus S$. Observe that there are only $\FPT$-many assignments for which the minimal weak $\calC$-backdoors have to be computed and as the satisfiability test for the formulas is in $\PTime$ this yields a $\delayFPT$ algorithm.
%
\end{proof}

Let $\varphi$ be a CNF-formula and $V\subseteq\textit{Vars}(\varphi)$ be a subset of its variables. Then $\varphi-V$ is the formula where all literals over variables from $V$ have been removed from all clauses in $\varphi$. 
Now we want to consider strong $\mathcal C$-backdoors for clause-defined base classes $\mathcal C$. Note that in this case the notion of strong $\mathcal C$-backdoors coincides with the notion of deletion $\mathcal C$-backdoors, i.\,e., a set $V\subseteq\textit{Vars}(\varphi)$ is a strong $\calC$-backdoor of $\varphi$ if and only if $\varphi-V\in\calC$. 

\begin{theorem}
 For every clause-defined base class $\mathcal C$ and input formulas in 3CNF, $\problemFont{Min-Strong-$\mathcal C$-Backdoors}\in\totalFPT$.
\end{theorem}
\begin{proof}
 By the previous observation we get that we only need to branch on the variables from a clause $C\notin\calC$ and remove the corresponding variable from $\varphi$. The size of the branching tree is at most $3^k$. As for base classes the satisfiability test is in $\PTime$ this yields a $\totalFPT$ algorithm.
\end{proof}


\begin{theorem}
For every clause-defined base class $\mathcal C$ and input formulas in 3CNF, $\problemFont{All-Strong-$\mathcal C$-Backdoors}\in\delayFPT$. 
\end{theorem}

\begin{proof}
Again we focus on the neighbourhood function $\nb{}(x,S)$ for $x=(\varphi,k)$ which is defined to be the set of the pairwise unions of all minimal strong $\calC$-backdoors of $(\varphi-(S\cup\{x_{i}\}),k-|S|-1)$ together with $S\cup\{x_i\}$ for all variable $x_{i}\not\in S$.
%
\end{proof}

\subsection{Weighted Satisfiability Problems}

We consider formula classes defined over the Schaefer framework defined over a set $\Gamma$ of constraints, and investigate the problem class $\minonesSAT(\Gamma)$ 
with respect to ordering by weight of assignments. These problems have been classified without a specific ordering in \cite{cmmsv13}. Given a propositional formula $\varphi$ and an assignment $\theta$ over the variables in $\varphi$ with $\theta\models\varphi$, we say $\theta$ is \emph{minimal} if there does not exist an assignment $\theta'$ which sets strictly less variables to true than $\theta$ and $\theta'\models\varphi$. The size $|\theta|$ of $\theta$ is the number of variables it sets to true. Formally, the problems from above are defined as follows:
\enumproblem
{$\Min\minonesSAT(\Gamma)$}
{$(\varphi,k)$, a propositional $\Gamma$-formula $\varphi$, $k\in\N$.}
{$k$}
{Generate all minimal assignments $\theta$ of $\varphi$ with $|\theta|\leq k$ by non-decreasing size.}

Similarly, the problem $\All\minonesSAT(\Gamma)$ asks for all assignments $\theta$ of $\varphi$ with $|\theta|\leq k$. 

\begin{theorem}
 For all constraint languages $\Gamma$, we have: $\Min\minonesSAT(\Gamma)\in\totalFPT$ and $\All\minonesSAT(\Gamma)\in\delayFPT$.
\end{theorem}
\begin{proof}
 For the first claim we can simply compute the minimal assignments by a straight forward branching algorithm: start with the all $0$-assignment, then consider all unsatisfied clauses in turn and flip one of the occurring variables to true. The second claim follows by a direct application of \Cref{thm:main}.
\end{proof}

%
%

\section{Conclusion}

We presented $\FPT$-delay ordered enumeration algorithms for a large variety of problems, such as cluster editing, chordal completion, closest-string, and weak and strong backdoors. An important point of our paper is that we propose a general strategy for efficient enumeration. This is rather rare in the literature, where usually algorithms are devised individually for specific problems. In particular, our scheme yields $\delayFPT$ algorithms for all graph modification problems that are characterized by a finite set of forbidden patterns.

The focus of the present paper was on graph-theoretic problems. A point we did not consider here is the introduction of new vertices as an operation. In the full version of the paper we will see that our generic approach covers this situation as well. Furthermore, we will define general modification problems, detached from graphs. We will introduce a general notion of \emph{modification operations} and provide generic enumeration algorithms for arising problems in the world of graphs, strings, numbers, formulas, constraints, etc.

As an observation we would like to mention that the $\delayFPT$ algorithms presented in this paper require exponential space due to the inherent use of the priority queues to achieve \emph{ordered} enumeration. An interesting question is whether there is a method which requires less space but uses a comparable delay between the output of solutions and still obeys the underlying order on solutions.

\bibliographystyle{plain}
\bibliography{bst}

\end{document}